\documentclass{article}

\usepackage{hyperref,amssymb,amsmath,graphicx,verbatim,amsthm}
\usepackage{color}

\newtheorem{theorem}{Theorem}
\newtheorem{cor}[theorem]{Corollary}
\newtheorem{lemma}[theorem]{Lemma}
\newtheorem{prop}[theorem]{Proposition}

\newtheorem*{remark*}{Remark}
\newtheorem*{claim*}{Claim}
\newtheorem*{theorem*}{Theorem}

\newcommand{\re}{\Re }

\newcommand{\Z}{\mathbb{Z}}
\newcommand{\N}{\mathbb{N}}
\newcommand{\Q}{\mathbb{Q}}
\newcommand{\R}{\mathbb{R}}

\newcommand{\C}{\mathbb{C}}

\newcommand\Arg{\hbox{Arg}}
\newcommand{\E}{\displaystyle \mathop{\mathbb E}}
\newcommand{\EE}{\mathop{\mathbb E}}

\newcommand\Newt{\hbox{Newt}}

\begin{document}
\title{On the distribution of runners on a circle}

\author{Pavel Hrube\v{s}\footnote{Institute of Mathematics of ASCR, Prague, pahrubes@gmail.com.  Supported by the GACR grant 19-27871X and 19-05497S.}}

\maketitle

\begin{abstract}  Consider $n$ runners running on a circular track of unit length with constant speeds such that $k$ of the speeds are distinct. We show that, at some time, there will exist a sector $S$ which contains at least  $|S|n+ \Omega(\sqrt{k})$ runners. 
The bound is asymptotically tight up to a logarithmic factor.  
The result can be generalized as follows. Let $f(x,y)$ be a complex bivariate polynomial whose Newton polytope has $k$ vertices. Then there exists $a\in \C\setminus\{0\}$ 
and a complex sector $S=\{re^{\imath \theta}: r>0, \alpha\leq \theta \leq \beta\}$ such that the univariate polynomial $f(x,a)$ contains at least $\frac{\beta-\alpha}{2\pi}n+\Omega(\sqrt{k})$ non-zero roots in $S$ 
(where $n$ is the total number of such roots and $0\leq (\beta-\alpha)\leq 2\pi$). 
This shows that the Real $\tau$-Conjecture of Koiran \cite{Koiran} implies the conjecture on Newton polytopes of Koiran et al. \cite{KoiranNewt}.   
\end{abstract}

\section{Introduction}\label{sec:intro}
Consider $n$ runners running on a circular track with constant and distinct speeds. Does it have to be the case that, at some point in time, they concentrate in some non-trivial sector?
If the speeds are sufficiently independent, Kronecker's approximation theorem \cite{Kronecker} implies that the runners will all meet in an arbitrarily small sector. On the other hand, if the speeds are $1,2,\dots,n$, it is easy to set the starting positions so that the runners never meet in a common half-circle, or any constant fraction of the circle. 
A similar construction  can also be deduced from the approximation theorem of Dirichlet \cite{Dirichlet}. Furthermore, if the starting positions are chosen randomly, the runners will be almost uniformly distributed at any point in time (see Section \ref{sec:run2} below).  Nevertheless, we will show that some deviation from uniformity must occur:

\begin{theorem*} Assume that $n$ runners run on a circle of unit length with constant speeds such that $k$ of the speeds are distinct. Then there exists a time and a sector $S$ such that $S$ contains at least $|S|n+ c\sqrt{k}$ runners, where $c>0$ is an absolute constant. 
\end{theorem*}         

Observe that $|S|n$ is the expected number of runners in $S$, had they been distributed uniformly, and the theorem asserts that at some time, the true distribution of runners is $c\sqrt{k}$-far from uniform. We will also show that the bound is asymptotically tight up to a logarithmic factor.

The problem of runners has an interesting application to distribution of roots of complex polynomials. Take a bivariate polynomial $f(x,y)$ such that its Newton polytope has $k$ vertices:  $\sum_{i=0}^{k-1} x^i y^{i^2}$ is an iconic example. 
Given $a\in \C\setminus\{0\}$, consider the univariate polynomial $f(x,a)$. Then the Theorem can be generalized as follows: there exists $a\in \C\setminus\{0\}$ so that if $r_1e^{2\imath\pi\alpha_1},\dots, r_ne^{2\imath\pi \alpha_n}$ are the non-zero roots of $f(x,a)$, the distribution of $\alpha_1,\dots, \alpha_n$ is $c\sqrt{k}$-far from uniform.

In the iconic example, a stronger and simpler result follows from a theorem of  Hutchinson \cite{Hutchinson}, see Section \ref{sec:statement} for a detailed discussion. This also gives one motivation for this problem.   A different motivation comes from the complexity of  algebraic computations. In \cite{Koiran},  Koiran has conjectured the following: if a \emph{univariate} polynomial $f(x)$ is sufficiently easy to compute then $f(x)$ has a small number of distinct real roots. This is called the Real $\tau$-Conjecture; it is rooted partly in Valiant's VP vs. VNP problem \cite{compValiant,BurgHabil}, partly in the $\tau$-Conjecture of Shub and Smale \cite{Smale}. Later, Koiran et al. \cite{KoiranNewt} conjectured that a similar statement holds for a \emph{bivariate} polynomial and the number of vertices of its Newton polytope. While the two conjectures seem related, and they share the crucial consequence that $\hbox{VP} \not= \hbox{VNP}$, no implication between them was previously known.  We can now conclude that in fact, the Real $\tau$-Conjecture implies the conjecture on Newton polytopes.

\section{Statement and discussion of main results} \label{sec:statement}
We first give the usual definition of the \emph{discrepancy} of a sequence. For $r_1,\dots, r_n \in [0,1]$,  
\[ D(r_1,\dots, r_n):= \sup_{0\leq a\leq b \leq 1} \frac{1}{n} | N_{a,b}(r_1,\dots, r_n) - n(b-a) |\,,\] 
where $N_{a,b}:= |\{i : r_i \in [a,b]  \}| $ is the number of $r_i$'s in $[a,b]$. 
For general $r_1,\dots, r_n\in \R$, we let $D(r_1,\dots, r_n):= D(\{r_1\},\dots, \{r_n\})$, 
where $\{r\}:= r- \lfloor r\rfloor$
is the fractional part of $r$. In our setting, the normalization factor $\frac{1}{n}$ in $D$ is  rather inconvenient, and we define the \emph{bias} of $r_1,\dots, r_n$ as
\[B(r_1,\dots, r_n):= nD(r_1,\dots, r_n)\,.\]
Our main theorem about runners can be formally restated\footnote{In the Introduction, we asserted that there is a sector which contains \emph{more} than the expected number of runners, whereas here we claim the existence of a sector containing more or \emph{less} than the expected number. But if a sector contains few runners, its complement must contain many; and we can keep it closed by enlarging it by $\epsilon$. } as follows:

\begin{theorem}\label{thm:run0} Let $s_1,\dots, s_n \in [0,1)$. Let $v_1,\dots, v_n\in \R$ and $k:= |\{v_1,\dots,v_n\}|$ be the number of distinct $v_i$'s.
Then there exists $t\in \R$ such that $B(s_1+v_1t,\dots, s_n+v_nt)\geq \sqrt{k/12}$. 
\end{theorem} 

The theorem will be proved in Section \ref{sec:runners1}, where  we also give a stronger result for $k=n$. In Section \ref{sec:run2}, we show that the bound in Theorem \ref{thm:run0} is tight up to a factor of $\sqrt{\log k}$. 

Let $f(x)\in \C[x]$ be a complex univariate polynomial. Assume that $f$ has $n$ (not necessarily distinct) \emph{non-zero} roots
$r_1e^{\imath \phi_1},\dots, r_ne^{\imath \phi_n}$, where $r_1,\dots r_n>0$ and $\phi_1\dots,\phi_n\in [0,2\pi)$. We define the \emph{bias of $f$}
\[B(f):= B\left (\frac{\phi_1}{2\pi},\dots, \frac{\phi_n}{2\pi}\right)\,. \]
Denoting $N_{\alpha,\beta}(f)$ the number of roots of $f$ in the complex sector 
$\{re^{\imath \theta}: r>0, \theta\in [\alpha, \beta]\}$, one can also write 
\[ B(f) =\sup_{0\leq \alpha\leq \beta\leq 2\pi}|N_{\alpha,\beta}(f)-\frac{n(\beta-\alpha)}{2\pi}|\,.\] 
We remark that $B(f)$ has been studied already in the classical paper of Erd\"os and Tur\'an \cite{ErdosTuran}. 

Let $f(x,y)= \sum_{i,j} a_{i,j}x^iy^j$ be a \emph{bivariate} complex polynomial. Let $\hbox{supp}(f):= \{(i,j): a_{i,j}\not=0\}\subseteq \Z^2$ be the set of exponents of monomials with a non-zero coefficient. The \emph{Newton polytope of $f$}, $\Newt(f)\subseteq \R^2$, is defined as the convex hull of $\hbox{supp}(f)$. In Section \ref{sec:NewtProof}, we will prove:

\begin{theorem}\label{thm:Newt}\label{thm:Newton} Let $f(x,y)$ be a bivariate complex polynomial such that $\Newt(f)$ has $k$ vertices. Then there exists $a\in \C\setminus \{0\}$ such that the univariate polynomial $f(x,a)$ satisfies\footnote{In this paper, $g(k)\geq \Omega(h(k))$ means that $g(k)\geq c \cdot h(k)$ holds for some constant $c>0$ and every sufficiently large $k$.} 
$B(f(x,a))\geq \Omega(\sqrt{k})$.
\end{theorem}  
The theorem can be motivated by the following example. Take a real polynomial $f(x)=\sum_{i=0}^n a_i x^i$. A theorem of Hutchinson \cite{Hutchinson}, which appears more explicitly   in \cite{Kurtz}, gives the following: if $a_i$ are positive and \begin{equation}
a_i^2> 4a_{i-1}a_{i+1}\label{eq:kurtz}\end{equation} for every $i\in \{1,\dots, n-1\}$ then all the roots of $f(x)$ are distinct, real and negative. 
Now consider the bivariate polynomial $g(x,y)=\sum_{i=0}^n b_i x^iy^{i^2}$ with $b_i>0$. Then we can set $a>0$ small enough, so that the coefficients of the univariate polynomial $g(x,a)$ satisfy (\ref{eq:kurtz}), and hence all the roots are real and negative. In the language of Theorem \ref{thm:Newt}, $B(g(x,a))= n$. In this argument, $i^2$ could be replaced by any strictly convex function (or strictly concave, letting $a\rightarrow \infty$). Furthermore, using a result of Karpenko and Vishnyakova \cite{Karpenko}, we can also assume that $b_i\in \R$ are non-zero (rather than positive), giving that the roots are real (rather than negative). However, things get more complicated if some coefficients are zero. In this case, we can no longer expect all the roots of $g(x,a)$ to be real, or lie on the same line $re^{\imath \phi}, r\in \R$. Theorem \ref{thm:Newt} nevertheless tells us that for some $a\in \C\setminus \{0\}$, the roots of $g(x,a)$ non-trivially concentrate in some complex sector.   

It is easy to see that Theorem \ref{thm:Newt} does not hold if $f(x,y)$ and $a$ are required to be real. 
Furthermore, as noted in Proposition \ref{prop:Newt}, the bound in Theorem \ref{thm:Newt} is tight up to $\sqrt{\log k}$ factor. 

We shall also give the following modification of Theorem \ref{thm:Newt}. For a univariate complex polynomial $f(x)=\sum a_i x^i$, let $\re(f)$ denote the real polynomial $\sum_i \re{(a_i)}x^i$ (where $\re(a_i)$ is the real part of $a_i$).

\begin{theorem} \label{thm:cor} Let $f(x,y)$ be a complex  polynomial such that its Newton polytope has $k$ vertices. Then there exists $a\in \C\setminus\{0\}$ such that
 $\re(f(x,a))$ has $\Omega(k)$ distinct real roots.  
 \end{theorem}
 
 This is proved in Section \ref{sec:cor}. We remark that a weaker bound of $\Omega(\sqrt{k})$ follows  from Theorem \ref{thm:Newt} and Cauchy's argument principle. 
 
 \paragraph{An application to Real $\tau$-Conjectures}
 The Real $\tau$- Conjecture of Koiran \cite{Koiran} asserts the following: let $f\in \R[x]$ be a real univariate polynomial which can be written as  \begin{equation} f= \sum_{i=1}^p \prod_{j=1}^qf_{i,j}\,,\,\, \hbox{ where } |\hbox{supp}(f_{i,j})|\leq r\,, \label{eq:koiran}\end{equation} then $f$ has at most $(pqr)^c$ distinct real roots (for some absolute constant $c$).   
In \cite{KoiranNewt}, Koiran et al. have  made a similar conjecture (called the $\tau$-Conjecture for Newton Polygons): let $f(x,y)$ be a real \emph{bivariate} polynomial as in (\ref{eq:koiran}), then $\Newt(f(x,y))$ has at most $(pqr)^{c'}$ vertices.  
Using Theorem \ref{thm:Newt}, we can conclude the two conjectures are related: 

\begin{prop}\label{prop:tau} The Real $\tau$-Conjecture implies the $\tau$-Conjecture for Newton Polygons.  
\end{prop}

\begin{proof} In  \cite{ja:tau}, it was shown that the Real $\tau$-Conjecture implies the following: given a \emph{complex} univariate polynomial $f$ as in (\ref{eq:koiran}), its bias  can also bounded as  $B(f)\leq (pqr)^{c''}$ (where $c''>0$ is a new absolute constant). Assume now that $f(x,y)$ is a (real or complex) polynomial of the form (\ref{eq:koiran}) such that $\Newt(f(x,y))$ has $k$ vertices. By Theorem \ref{thm:Newt}, we can find $a\in \C\setminus \{0\}$ so that $B(f(x,a))\geq C\sqrt{k}$, with a constant $C>0$. Assuming the Real $\tau$-Conjecture,  the result in \cite{ja:tau} gives that $C\sqrt{k}\leq (pqr)^{c''}$ and hence $k\leq C^2(pqr)^{2c''}$ -- as required\footnote{As the non-trivial case is $pqr\geq 2$, the constant $C^2$ can be subsumed in the exponent.} by the conjecture on Newton polytopes.  
\end{proof}

We point out that the same could be concluded from Theorem \ref{thm:cor} and a lemma from \cite{ja:tau} relating the complexity of $f$ with that of $\re(f)$.

\paragraph{Some notation} For  $n\in \N$, let $[n]:=\{1,\dots, n\}$. For $r\in \R$, $\{r\}:= r-\lfloor r\rfloor$ is the fractional part of $r$. $\log(x)$ is the logarithm in base two.

\section{Lower bounds on the discrepancy of runners}\label{sec:runners1}
In this section, we prove Theorem \ref{thm:run0}, as well as give a stronger bound in the special case $k=n$. 
We point out that the special case can be easily proved by estimating $\max_{t\in [0,1]}|\sum_{j=1}^ne^{2\pi\imath(s_j+v_j t)}|$ and then using some well-known properties of discrepancy, such as the Koksma inequality (see, e.g., the monograph \cite{Discrepancy}). In the case of non-distinct speeds, this approach seems hard to implement and leads to 
Tur\'an style problems on power-sums \cite{TuranMethod,TuranInequalities}.  
The strategy of our proof is therefore different. We directly estimate the expectation of the square of the number of runners in $S$, for a random time $t$ and a random sector $S$. 
It is more convenient to first analyze the case when the speeds are integers: at time $t=1$ the runners return to their original positions and it is enough to understand the system in the interval $t\in [0,1]$.  


Given $0\leq \gamma\leq 1$ and $\alpha\in \R$, let
\[ S_{\alpha,\alpha+\gamma}:= \{ x\in [0,1]: \{x-\alpha\}\leq \gamma\}\,. \]
When $[0,1]$ is viewed as a circle, $S_{\alpha,\alpha+\gamma}$  is the closed sector which starts at $\alpha$ and continues clockwise for distance $\gamma$. $\gamma$ will be called the \emph{aperture} of $S_{\alpha,\alpha+\gamma}$ and denoted $|S_{\alpha,\alpha+\gamma}|$. Given a sector $S$ and $x\in \R$, we define
\[\chi_{S}(x)=\begin{cases} ~~1\,, &  \{x\}\in S\\
 ~~0\,,& \{x\}\not \in S
 \end{cases}\,.\]  
Let $N_S(x_1,\dots,x_n):= \sum_{i=1}^n \chi_S(x_i)$ and $\chi_{\alpha,\gamma}(x):= \chi_{S_{\alpha,\alpha+\gamma}}(x)$.  

\begin{remark*} \[
B(r_1,\dots,r_n)= \sup_{0\leq \alpha,\gamma \leq 1} |N_{S_{\alpha,\alpha+\gamma}}(r_1,\dots, r_n)- \gamma n|\,.\]
Moreover, it does not matter whether the sectors are closed, open, or half-open.
\end{remark*}

In the following, $\EE_x h(x)$ will stand for $\int_{0}^1 h(x)dx$, the expectation of $h(x)$ on $[0,1]$. Similarly, $\EE_{x,y}h(x,y)$ stands for $\int_{0}^1\int_{0}^1h(x,y)dydx$ etc. 
In the cases considered below, Fubini's theorem is applicable and we have $\EE_{x,y}=\EE_{y,x}$.   


\begin{lemma}\label{lem:run1}
Let $s_1,s_2\in \R$, let  $v_1,v_2$ be distinct integers and $\gamma\in [0,1]$. Then
\begin{eqnarray}\E_{t,\alpha} \chi_{\alpha,\gamma}(s_1+v_1t)&=&\gamma\,,\label{eq:1}\\
\E_{t,\alpha} \chi_{\alpha,\gamma}(s_1+v_1t)\chi_{\alpha,\gamma}(s_2+v_2t)&=&\gamma^2\,\label{eq:2}\,.
\end{eqnarray}
\end{lemma}

\begin{proof} (\ref{eq:1}) is rather obvious. In fact, we already have
\[\E_\alpha \chi_{\alpha,\gamma}(s+vt)=\gamma\,,\,\, \E_{t} \chi_{\alpha,\gamma}(s+vt)=\gamma\,,\]
where the latter holds if\footnote{The assumption $v\in \Z$ makes $\chi_{\alpha,\gamma}(s+vt)$ 1-periodic in $t$, and $v\not=0$ guarantees the runner spends $\gamma$-fraction of time in $S_{\alpha,\alpha+\gamma}$.} $v\in \Z\setminus \{0\}$.  

To prove (\ref{eq:2}), 
note that $\chi_{\alpha,\gamma}(x+z)= \chi_{\alpha-z,\gamma}(x)$ and the left-hand side of (\ref{eq:2})
equals $A:= \int_{0}^1\int_{0}^1\chi_{\alpha-v_1t,\gamma}(s_1)\chi_{\alpha-v_2t,\gamma}(s_2)d\alpha d t$. We have
\begin{align*}
\int_{0}^1\chi_{\alpha-v_1t,\gamma}(s_1)\chi_{\alpha-v_2t,\gamma}(s_2)d\alpha=& \int_{-v_1t}^{1-v_1t}\chi_{\alpha',\gamma}(s_1)\chi_{\alpha'+(v_1-v_2)t,\gamma}(s_2)d\alpha'=\\ 
=&\int_{0}^{1}\chi_{\alpha',\gamma}(s_1)\chi_{\alpha'+(v_1-v_2)t,\gamma}(s_2)d\alpha' \,,
\end{align*}
where we have used the substitution $\alpha'= \alpha-v_1t$ and the fact that $\chi_{\alpha,\gamma}$ is 1-periodic in the first argument. 
Exchanging the order of integration, we have 
\begin{align*}A=&\int_{0}^1\int_{0}^{1}\chi_{\alpha,\gamma}(s_1)\chi_{\alpha+(v_1-v_2)t,\gamma}(s_2)d\alpha dt=\\
= &\int_{0}^{1}\left(\chi_{\alpha,\gamma}(s_1)\int_{0}^1\chi_{\alpha+(v_1-v_2)t,\gamma}(s_2)d t\right)d\alpha= 
\int_{0}^{1}\left(\chi_{\alpha,\gamma}(s_1)\gamma \right)d\alpha=\gamma^2\,.
\end{align*}
\end{proof}

As a warm-up for Theorem \ref{thm:run0}, we first consider the case of runners with distinct speeds. In this case, the obtained result is stronger. 

\begin{theorem}\label{thm:run1} Let $v_1,\dots, v_n$ be distinct real numbers and $s_1,\dots, s_n\in [0,1)$. Let $\gamma \in [0,1]$. Then there exists $t\in \R$ and a sector $S=S_{\alpha, \alpha+\gamma}$ of aperture $\gamma$ such that \[|N_S(s_1+v_1t, \dots, s_n+v_n t)- \gamma n|\geq \sqrt{(\gamma-\gamma^2)n}\,.\] In particular, there exists $t\in \R$ such that $B(s_1+v_1t, \dots, s_n+v_n t)\geq \sqrt{n}/2$. 
\end{theorem}

\begin{proof} Assume first that $v_1,\dots, v_n$ are distinct integers. 
Define $N(\alpha,t):= \sum_{i=1}^n \chi_{\alpha, \gamma}(s_i+v_i t)$; the number of runners in $S_{\alpha,\alpha+\gamma}$ at time $t$.
We shall abbreviate $\EE_{\alpha,t}$ by $\EE$.  We want to estimate $\E((N(\alpha,t)-\gamma n)^2)$. By the previous lemma,
we obtain $\E(N(\alpha,t))= \gamma n$. This implies $\E((N(\alpha,t)-\gamma n)^2)= \E((N(\alpha,t)^2)-\gamma^2 n^2$. 
Moreover, the lemma also gives \begin{eqnarray*}\E(\chi_{\alpha, \gamma}(s_i+v_it)\chi_{\alpha, \gamma}(s_j+v_jt))&=& \gamma\,, \hbox{ if } i=j\,, \\
&=& \gamma^2\,, \hbox{ if } i\not=j \,.\end{eqnarray*}
This means that 
\begin{align*}\E (N(\alpha,t)^2)=& \sum_{i=1}^n \E(\chi_{\alpha,\gamma}(s_i+v_it)^2)+ \sum_{i\not =j }\E(\chi_{\alpha,\gamma}(s_i+v_i t)\chi_{\alpha,\gamma}(s_j+v_jt))=\\
=& \gamma n+\gamma^2 n(n-1)=  (\gamma-\gamma^2)n + \gamma^2n^2\,.\end{align*}
Altogether, we obtain 
\[\E((N(\alpha,t)-\gamma n)^2)=(\gamma-\gamma^2)n + \gamma^2n^2-\gamma^2n^2= (\gamma-\gamma^2)n\,. \]
This means that for some $\alpha$ and $t$, $|N(\alpha,t)- \gamma n|\geq \sqrt{(\gamma-\gamma^2)n}$ which proves the special case of the theorem. 

For non-integer speeds, we will apply Dirichlet's approximation theorem. If $v_1',\dots, v_n'$ are distinct real numbers and $\epsilon>0$, the theorem gives a positive integer $q$ and integers $v_1,\dots, v_n$ such that $|v'_iq-v_i|\leq \epsilon$ for every $i$. Let $N'(\alpha,t):= \sum_{i=1}^n \chi_{\alpha, \gamma}(s_i+v_i' t)$ and $N(\alpha,t)$ be as above. We again want to estimate  $|N'(\alpha,t)-\gamma n|$. Since we can scale the time by a factor of $q$, we can assume that in fact $q=1$ and so $|v_i'-v_i|\leq \epsilon$ for every $i$. We observe that 
\[\E((N'(\alpha,t)-\gamma n)^2) \geq \E((N(\alpha,t)-\gamma n)^2) - \epsilon c_n\,,  \]
where $c_n$ is a constant depending only on $n$. This is because 
$N(\alpha,t)$ and $N'(\alpha,t)$ differ on at most an $2\epsilon n$-fraction of $\alpha\in [0,1]$. 
Hence we conclude that there exists $\alpha$ and time $t$ with $|N'(\alpha,t)- \gamma n|\geq \sqrt{(\gamma-\gamma^2)n}-\epsilon c'_n$. Since we can pick $\epsilon$ arbitrarily small and $N'(\alpha,t)$ is an integer, we conclude $|N'(\alpha,t)- \gamma n|\geq \sqrt{(\gamma-\gamma^2)n}$ for some $\alpha,t$.  
\end{proof}

We remark that the main part of Theorem \ref{thm:run1} fails miserably if the runners have non-distinct speeds. 
Consider $n=2k$ runners with $k$ distinct speeds running in pairs such that in a given pair, the two runners maintain distance $1/2$. Then we can set the starting positions so that for every sector $S$ of aperture $\gamma=1/2$ and every time, the number of runners in the sector is at least $k$ and at most\footnote{This is owing to the fact that $S$ is closed; a half-closed sector would contain precisely $k$ runners at every time.} $k+4$.     

To prove Theorem \ref{thm:run0}, we need one more lemma:

\begin{lemma}\label{lem:run2}
Let $s_1,\dots,s_m \in \R$ and $\alpha\in [0,1]$. Let $N(\gamma):= \sum_{i=1}^m\chi_{\alpha,\gamma}(s_i)$.  Then \[\E_{\gamma}((N(\gamma)-\gamma m)^2)\geq \frac{1}{12}\,. \]
\end{lemma}

\begin{proof} The function $N(\gamma)$ is integer-valued. Hence $|N(\gamma)-\gamma m|\geq \Delta(\gamma m)$, where $\Delta(z)\in [0,1/2]$ denotes the distance of $z\in \R$ from a closest integer. It is therefore enough to estimate $\int_{0}^1 \Delta(\gamma m)^2d\gamma$. 
The function $\Delta(\gamma m)$ is $1/m$-periodic and symmetric with respect to the point $\gamma_0=1/2m$. 
This means that $\int_{0}^1 \Delta(\gamma m)^2d\gamma= 2m \int_{0}^{1/2m} \Delta(\gamma m)^2d\gamma$. Furthermore, if $\gamma\in [0,1/2m]$ then $\Delta(\gamma m)=\gamma m$ and  hence 
\[\int_{0}^{1/2m} \Delta(\gamma m)^2d\gamma= m^2\int_{0}^{1/2m} \gamma^2d\gamma= \frac{m^2}{3(2m)^3}= \frac{1}{24m}\,.\]
This gives  $\int_{0}^1 \Delta(\gamma m)^2\geq 2m/24m=1/12$. 
\end{proof}


\begin{proof}[{\bf Proof of Theorem \ref{thm:run0}}] We will assume that $v_1,\dots, v_n$ are integers; the general case proceeds in the same way as in the proof of  Theorem \ref{thm:run1}. Without loss of generality, assume that already $v_1,\dots,v_k$ are distinct. Given $j\in [k]$, let $A_j:=\{i\in [n]: v_i=v_j\}$ be the set of runners with speed $v_j$. Let 
$N_j(\alpha,\gamma,t):=  \sum_{i\in A_j} \chi_{\alpha, \gamma}(s_i+v_i t)$ and $N(\alpha,\gamma, t):= \sum_{i=1}^n \chi_{\alpha, \gamma}(s_i+v_i t)$. Hence $N(\alpha,\gamma, t)$ denotes the number of runners in $S_{\alpha,\alpha+\gamma}$ at time $t$ and $N(\alpha,\gamma,t)=\sum_{j=1}^kN_j(\alpha,\gamma,t)$. 
We want to estimate 
$\E((N(\alpha,\gamma,t)-\gamma n)^2)$, where $\E$ now stands for $\EE_{\alpha,\gamma,t}$.   

Setting $g_j(\alpha,\gamma, t):= N_j(\alpha,\gamma, t)- \gamma |A_j|$, we claim that \begin{align}
\E(g_j^2)&\geq 1/12\,,\,\, \nonumber \\ 
\E(g_{j_1}g_{j_2})&=0, \,\hbox{ if } j_1\not=j_2\,.
\label{eq:B}
\end{align}
The first inequality is a consequence of Lemma \ref{lem:run2}. (\ref{eq:B}) is an application of Lemma \ref{lem:run1} as follows. 
For a fixed $\gamma$, we have $\EE_{\alpha,t} N_j(\alpha,\gamma,t)= \gamma|A_j|$ which means that 
\[\E_{\alpha,t}(g_{j_1}g_{j_2})= \E_{\alpha,t}(N_{j_1}(\alpha,\gamma,t)N_{j_2}(\alpha,\gamma,t))-\gamma^2|A_{j_1}||A_{j_2}|\,.\]
Furthermore, by Lemma \ref{lem:run1},
\begin{align*}\E_{\alpha,t}(N_{j_1}(\alpha,\gamma,t)N_{j_2}(\alpha,\gamma,t))=& \sum_{i_1\in A_{j_1}, i_2\in A_{j_2}}\E_{\alpha,t}\chi_{\alpha,\gamma}(s_{i_1}+v_{j_1}t)\chi_{\alpha,\gamma}(s_{i_2}+v_{j_2}t)=\\ =& \gamma^2 |A_{j_1}||A_{j_2}|\,.\end{align*} This shows that the left-hand side of (\ref{eq:B}) indeed equals zero. 

We now have 
\begin{align*} \E((N(\alpha,\gamma,t)-\gamma n)^2)=&\E((\sum_{j=1}^k g_j)^2)=
 \sum_{j=1}^k\E(g_j^2)+\sum_{j_1\not =j_2}\E(g_{j_1}g_{j_2})
\geq  \frac{k}{12}\,.
\end{align*}
This implies that for some $\alpha, t, \gamma$, $|N(\alpha,\gamma,t)-\gamma n|\geq \sqrt{k/12}$. 
\end{proof}

\section{An upper bound on the discrepancy of runners}\label{sec:run2}
We now want to show that the bounds in Theorem \ref{thm:run0} and Theorem \ref{thm:run1} are tight up to logarithmic factors.

\begin{theorem}\label{thm:run2} Let $v_1:=1,\dots, v_n:=n$. There exist $s_1,\dots, s_n\in [0,1)$ such that for every $t\in \R$ and every $\gamma \in [0,1]$, the following holds. 
For every sector $S$ of aperture $\gamma$, $|N_S(s_1+v_1t,\dots, s_n+v_nt) -\gamma n| \leq O(\sqrt{n\gamma \log n}+\log n)$. Hence 
$B(s_1+v_1t,\dots, s_n+v_nt)\leq O(\sqrt{n\log n})$.
\end{theorem}

Clearly, this implies a similar bound in the general case of non-distinct speeds.  If $k\leq n$, set $v_1,\dots, v_n$ so that $v_1=1,\dots, v_k=k $ and $v_{k+1},\dots, v_n=k$. Applying Theorem \ref{thm:run2} and setting the starting positions of the last $n-k+1$ runners so that they uniformly partition the circle, we have for every $t$
\[B(s_1+v_1t,\dots, s_n+v_nt)\leq O(\sqrt{k\log k})\,.\]

\begin{proof}[Proof of Theorem \ref{thm:run2}] 
Pick $s_1,\dots, s_n\in [0,1]$ uniformly and independently at random. Let $N_S(t):= N_S(s_1+v_1t,\dots, s_n+v_nt)$ be the number of runners in $S$ at time time $t$. 
We claim that for every fixed $t$ and a fixed sector $S$ of aperture $\gamma\geq 4 \log n/n$,
\begin{equation} \label{eq:chernoff} {\displaystyle \Pr_{s_1,\dots, s_n}}\left [ |N_S(t)-\gamma n| \geq 4\sqrt{n\gamma\log n}\right ]\leq n^{-5}\,.
\end{equation}
For if $s_1,\dots, s_n$ are uniform and independent, so are $s_1+v_1 t,\dots,s_n+v_nt$. The expected value of $N_S(t)$ is $\mu:= \gamma n$. Chernoff bound gives that for every $0\leq \delta \leq 1$,
\[\Pr[N_S(t)\leq (1-\delta)\mu] \leq e^{-\delta^2\mu/2}\,,\,\, \Pr[N_S(t)\geq (1+\delta)\mu] \leq e^{-\delta^2\mu/3}\,.\]
If we now set $\delta := 4\sqrt{\frac{\log n}{\gamma n}} $, we have $\delta\mu= 4\sqrt{n\gamma \log n}$ and $\delta^2\mu= 16\log n$. Hence both the probabilities in (\ref{eq:chernoff}) are at most $e^{-16\log n/3}\leq n^{-5}$. 

Let $m:= \lfloor n/(4 \log n)\rfloor$ and $\gamma_0:= 1/m$. Let ${\cal S}$ be the set of sectors of the form $S_{i\gamma_0, i\gamma_0+j\gamma_0}$, $i,j\in \{0,\dots, m-1\}$. That is, ${\cal S}$ consists of  the $m^2$ sectors  whose starting point and aperture is a multiple of $\gamma_0$. Let ${\cal T}$ be the set of times of the form $ k/ {nm}$, $k\in \{0,\dots, nm+1\}$. Since $m<n$, we have $|{\cal S}|\cdot |{\cal T}|\leq n^4$. Then (\ref{eq:chernoff}) and the union bound give that, with positive probability,  $\left|N_{S}- |S|\cdot n\right| \leq 4\sqrt{n|S|\log n}$ holds for every $S\in {\cal S}$ and $t\in {\cal T}$. 

Hence there exist $s_1,\dots, s_n$ so that
\begin{equation}\label{eq:forall} |N_{S}- |S|\cdot n| \leq 4\sqrt{n|S|\log n}\,, \,\, \hbox{ for all } S\in {\cal S}, t\in {\cal T}\,. \end{equation}
We first claim that this can be extended also to times not in ${\cal T}$: 
\begin{equation}\label{eq:forall2} |N_{S}- |S|\cdot n| \leq 4\sqrt{n|S|\log n}+O(\log n)\,, \,\, \hbox{ for all } S\in {\cal S}, t\in \R\,. \end{equation}
For, given a ''small" sector $S_0$ in $\cal S$ of aperture $\gamma_0$, (\ref{eq:forall}) tells us that $N_{S_0}\leq 8\log n+1$. Between two consecutive times $t_1=k/nm$ and $t_2=(k+1)/nm$ in $\cal T$, the fastest runner with speed $n$ covers distance $n/nm=\gamma_0$. This  means that the runners that come to or leave from a sector $S$ must come from, or move to, the at most two adjacent small sectors of aperture $\gamma_0$.
      In a similar fashion, we can extend (\ref{eq:forall2}) to all sectors of aperture $\gamma \in [0,1]$. For given such a sector $S$, we can find $S_1\in {\cal S}\cup\{\emptyset\}$ and $S_2\in\cal S$ with $S_1\subseteq S\subseteq S_2$ and apertures satisfying  $|S_1|\geq \gamma-2\gamma_0$, $|S_2|\leq \gamma+2\gamma_0$. 
\end{proof}

\subsection{An explicit construction}\label{sec:explicit}
It would be interesting to give an explicit construction of low-discrepancy runners, and we now make a step in this direction. We will use the  Erd\"os-Tur\'an inequality \cite{Erdos-Turan} which is a useful tool for bounding discrepancy. We also note that the inequality would somewhat simplify the proof of Theorem \ref{thm:run2} (see also \cite{ErdosRenyi}), at the cost of obtaining weaker bounds.

It is convenient to interpret the discrepancy of runners in terms of norms of complex polynomials. Let $f(x)=\sum_i a_i x^i$ be a complex polynomial. Let 
$|f|_m:= \max_{|x|=1}|f(x)|$
be its maximum on the unit complex circle. Furthermore, let 
$f^{(k)}(x):= \sum_i a_i^k x^i$
be the Hadamard power of $f$.  The following lemma is a straightforward adaptation of the Erd\"os-Tur\'{a}n inequality to our setting:

\begin{lemma}\label{lem:ET} Let $v_1,\dots,v_n$ be distinct non-negative integers, $s_1,\dots,s_n\in [0,1)$ and $f(x):= \sum_{i=1}^{n} e^{2\pi \imath  s_{i}} x^{v_i}$. Then for every $t\in \R$,
\[B(s_1+ v_1t,\dots, s_n+v_n t)\leq c\left(1+\sum_{k=1}^n \frac{|f^{(k)}(x)|_m}{k}\right)\,,\]
where $c>0$ is an absolute constant. 
\end{lemma}

In order to apply Lemma \ref{lem:ET}, we want to find a polynomial $f(x)$ with unimodular coefficients such that $|f^{(k)}(x)|_m$ is small for every $k\leq n$. 
   Our construction is a generalization of that of Shapiro polynomials, see, e.g., \cite{Shap}. Shapiro's construction  gives a polynomial $f(x)$ with $\pm 1$ coefficients and degree $d=2^n-1$ such that $|f(x)|_m\leq 2^{(n+1)/2}= \sqrt{2(d+1)}$.

Let us fix a prime $p$. Let $\xi$ be a $p$-th primitive root of unity. Let $D$ be the $p\times p$ (unnormalized) discrete Fourier transform matrix, $D_{j,k}= \xi^{jk}, j,k \in \{0,\dots, p-1\}$.
Recursively, we construct a $p$-tuple of polynomials $Q_{0,r},\dots, Q_{p-1,r}$. We set $Q_{0,0}:=1,\dots, Q_{p-1,0}:= 1$. If $r\geq 0$,
we let 
\begin{equation}\label{eq:shapiro}\left( \begin{array}{l} Q_{0,r+1}\\ Q_{1,r+1}\\ \vdots \\Q_{p-1,r+1}\end{array} \right) = D \cdot \left( \begin{array}{r} Q_{0,r}\\ x^{p^r}Q_{1,r}\\ \vdots \\x^{(p-1)p^r}Q_{p-1,r}\end{array} \right)\,.\end{equation}
The construction guarantees that every $Q_{i,r}$ has degree $d_r=p^r-1$ and that its coefficients have absolute value one.

For example, in the case $p=2$, we obtain the usual Shapiro polynomials. The definition is simplified to
\[Q_{0,0}\,,Q_{1,0}=1\,,\,\,\,\,\, \left( \begin{array}{l} Q_{0,r+1}\\ Q_{1,r+1}\end{array} \right) = \left(\begin{array}{l r} 1&1\\ 1&-1\end{array} \right) \cdot  \left(\begin{array}{r} Q_{0,r}\\ x^{2^r}Q_{1,r}\end{array} \right)\,,\]
and gives the sequence:
\[ \begin{array}{l} 1\\ 1\end{array} \,,\,\, \begin{array}{l} 1+x\\ 1-x\end{array} \,,\,\,  \begin{array}{l} 1+x+x^2-x^3\\ 1+x-x^2+x^3\end{array} \,, \,\,  \begin{array}{l} 1+x+x^2-x^3+ x^4+x^5-x^6+x^7\\ 1+x+x^2-x^3- x^4-x^5+x^6-x^7\end{array} \,,\, \dots\]

We can bound $|Q_{i,r}^{(k)}|_m$ as follows. 

\begin{prop}\label{prop:shapiro} Let $0\leq i\leq p-1$ and $k$ be a natural number such that $p$ does not divide $k$. Then $|Q_{i,r}^{(k)}|_m\leq p^{\frac{r+1}{2}}= \sqrt{p(d_r+1)}$. 
\end{prop}

  \begin{proof} Assume first that $k=1$. We will prove that for every $x$ with $|x|=1$,
  \begin{equation}\label{eq:shapiro2} |Q_{0,r}(x)|^2+\dots+|Q_{p-1,r}(x)|^2= p^{r+1}\,. \end{equation}
  This is by induction on $r$. If $r=0$, the statement is clear. For the inductive step, let $Q_{r+1}(x)$ be the vector on the left-hand side of (\ref{eq:shapiro}) and $P_r(x)$ the one on the right-hand side, 
  so that $Q_{r+1}(x)= DP_r(x)$. 
   For $u\in \C^p$ let $| u|$ be its Euclidean norm. The matrix $D$ satisfies $D\cdot \bar D^t = p I_p $.  This means that for every $u\in \C^p$, 
$|Du|^2=  p |u|^2$. 
 Hence we have $|Q_{r+1}(x)|^2= p|P_r(x)|^2$ for every $x\in \C$. Furthermore, if $|x|=1$, we have $|P_r(x)|^2= |Q_r(x)|^2$ and hence $|Q_{r+1}(x)|^2= p|Q_r(x)|^2$. This implies (\ref{eq:shapiro2}).  

Equality (\ref{eq:shapiro2}) gives $|Q_{i,r}|_m\leq p^{(r+1)/2}$ as required. Let $k$ be such that $p$ does not divide $k$. Then $Q^{(k)}_{i,r}$ satisfy the same recursive definition, 
except that the root $\xi$ is replaced with $\xi^k$, and the same conclusion holds.     
    \end{proof}
    
 In order to obtain low-discrepancy runners from Proposition \ref{prop:shapiro}, it is enough to take the polynomial $Q_{0,r}$ for a suitable $r$. It turns out that $r=3$ gives optimal parameters in this setting.\footnote{Hence the recursion (\ref{eq:shapiro}) is applied $3$ times. Note, however, that  $Q_{0,3}$ implicitly depends on the chosen prime $p$ (and the root $\xi$).}   

\begin{cor}\label{cor:shapiro} Let $p$ be a prime and $n:= p^3$. Let $s_0,\dots, s_{n-1}\in [0,1)$ be such that $Q_{0,3}=\sum_{j=0}^{n-1} e^{2\pi\imath  s_j}x^j$. Then 
$B(s_0,s_1+t, \dots, s_{n-1}+(n-1) t)\leq O(n^{2/3}\log n)$ for every $t\in \R$.
\end{cor}

\begin{proof} By Lemma \ref{lem:ET}, it is enough to estimate $A:= \sum_{k=1}^n \frac{|Q_{0,3}^{(k)}|_m}{k}$. If $p\nmid ~k$, we have $|Q_{0,3}^{(k)}|_m \leq p^2= n^{2/3}$ by Proposition \ref{prop:shapiro}. Hence 
\[\sum_{k\leq n, p\nmid ~k} \frac{|Q_{0,3}^{(k)}(x)|_m} {k}\leq \sum_{k=1}^n \frac{n^{2/3}}{k}\leq O(n^{2/3}\log n)\,. \]
If $p$ divides $k$, we have $|Q_{0,3}^{(k)}|_m \leq n$. Hence   
\[\sum_{k\leq n, p | k} \frac{|Q_{0,3}^{(k)}(x)|_m} {k}\leq  \sum_{k\leq n, p |k} \frac{n} {k}= \frac{n}{p} \sum_{a=1}^{n/p} \frac{1}{a}\leq O(n^{2/3}\log n)\,. \]
This gives the estimate $A\leq O(n^{2/3}\log n)$. 
\end{proof}

\section{Newton polytopes and angular distribution of zeros}

In this section, we prove Theorems \ref{thm:Newt} and \ref{thm:cor}. 

\subsection{The connection with runners}\label{sec:connection}
We start by discussing the connection between Theorem \ref{thm:Newt} and the discrepancy of runners. Let $f(x,y)$ be a polynomial
of the form
\begin{equation}\label{eq:poly} f(x,y) =  x^{m_1}y^{m_2}\prod_{i=1}^n(x-a_i y^{q_i} )\,,
\end{equation}
where $a_i= r_ie^{2\pi \imath  s_i}$, $r_i>0$, $s_i\in [0,1)$, and $m_1,m_2, q_1,\dots, q_n\in \Z$ with $m_1,m_2\geq 0$.   
Let $k:= |\{q_1,\dots, q_n\}|$ be the number of distinct $q_i$'s. 
Then $\Newt(f(x,y))$ has precisely $2k$ vertices (if $k>0$). This is because the Newton polytope of a product $g_1g_2$  is the Minkowski sum of Newton polytopes of the factors $g_1$ and $g_2$ (see, e.g., \cite{KoiranNewt} or references within). Hence, $\Newt(f(x,y))$ is the Minkowski sum of line segments (and a point) with precisely $k$ distinct gradients, which yields $2k$ vertices.     
Given $a= r e^{2\pi \imath  t }$, $r>0$, the non-zero roots of $f(x,a)$ are of the form $r^{q_j}r_j e^{2\pi\imath (s_j+ q_j t) }$, $j\in [n]$. Hence, as $t$ varies, their arguments 
are \[2\pi (s_1+q_1 t)\,, \dots, 2\pi (s_n+q_n t)\,, \]
and they can be seen as a system of runners on a circle of length one with speeds $q_1,\dots, q_n$. Using Theorem \ref{thm:run0}, these observations entail:
\begin{equation*}\label{eq:polyrun} B(f(x,a))\geq \Omega(\sqrt{k})\,,\,\hbox{for some } a \hbox{ with } |a|=1\,.
\end{equation*}
Conversely, a system of runners $s_1+v_1 t, \dots, s_n+v_n t$, $v_1,\dots,v_n\in \N$, can be associated with the bivariate polynomial 
\[ g(x,y)= \prod_{j=1}^n (x- e^{2\pi \imath  s_j} y^{v_j})\,.\]
This and Theorem \ref{thm:run2} implies: 

\begin{prop}\label{prop:Newt} For every $n$, there exists $g(x,y)$ whose Newton polytope has $2n$ vertices but for every $a\in \C\setminus \{0\}$,
$B(g(x,a))\leq O(\sqrt{n\log n})$. 
\end{prop} 

We remark that $g$ can be assumed to have \emph{real} coefficients by taking instead $ g\cdot \bar g= \prod_{j=1}^n (x^2- 2 \cos(2\pi s_j)  x y^{v_j} +y^{2v_j})$. 
Furthermore, Theorem \ref{thm:run2} gives more information: for example, every small sector of aperture $O(\log n/n)$ contains at most $O(\log n)$ runners. The same could be
said about the roots of $g(x,a)$.

\subsection{Proof of Theorem \ref{thm:Newt}} \label{sec:NewtProof}
Our goal is to deduce Theorem \ref{thm:Newt} from Theorem \ref{thm:run0}. The strategy is to approximate $f(x,y)$ by  polynomials corresponding to edges of its Newton polytope -- hence reducing the problem to the already understood case as in (\ref{eq:poly}). As pointed out by an anonymous referee, the following proof is similar to the proof of Newton-Puiseux theorem (see, e.g., \cite{Walker}). The theorem expresses roots of $f(x,y)=0$, when viewed as a polynomial in $x$, as Puiseux series in $y$. Moreover, the first approximation of the series is given by monomials on the boundary  $\Newt(f(x,y))$, which would lead to an alternative proof of Theorem \ref{thm:Newt} from Theorem \ref{thm:run0}.    

Let $g(x,y)$ be a polynomial such that $\Newt(g)$ lies on the line $\ell= \{(t, qt+m): t\in \R\}$, $q\in \Q$. Then $g(x,y)$ can be written as
\begin{equation}\label{eq:g} g(x,y)=  y^m \sum_{j=n_1}^{n_2} c_j x^{j}y^{q j}\,.
\end{equation}
Furthermore, if $q\in \Z$, $g(x,y)$ can be factored as  
\begin{equation}\label{eq:line} g(x,y)= a_0y^{m}\prod_{i=1}^n(xy^q-a_i )\,,
\end{equation} 
where $a_0,\dots, a_n\in \C$.  

\begin{lemma}\label{lem:rouche} Let $g(x,y)$ be as in (\ref{eq:line}) with $a_0\not=0$. Let $h(x,y)$ be a polynomial such that $\Newt(h)$ lies in the strict upper-half plane determined by $\ell$. Then for every $\epsilon>0$ sufficiently small, and every $a\in \C$ with $0< |a|$ sufficiently small with respect to $\epsilon$, the following holds.  Let $\xi$ be a non-zero root of $g(x,a)$ of multiplicity $p$. 
Then $g(x,a)+h(x,a)$ has precisely $p$ roots $\xi'$ which satisfy $|\xi'-\xi|\leq \epsilon |\xi|$, counted with multiplicity. 
\end{lemma}

\begin{proof} This is an application of Rouch\'e's theorem. Let $0< \epsilon < 1$ be such that $\epsilon< |a_i-a_j|$ for every $a_i\not=a_j$.  
Given $a\in \C\setminus \{0\}$, every root of $g(x,a)$ is of the form  $\xi(a)=a_ka^{-q}$ for some $a_k\not=0$. Let us fix such an $a_k$. 
Let $\Omega(a)$ be the open ball with centre at $\xi(a)$ and radius $\epsilon|\xi(a)|$. Let $M_1$ be the minimum of $|g(x,1)|$ on $\partial \Omega(1)$. Hence $M_1>0$ and  
 \begin{equation}\label{eq:min} \min_{x\in \partial \Omega(a)}|g(x,a)| \geq M_1 |a|^m\,.\end{equation}
Furthermore, we claim that there exists $M_2$ independent on $a$ such that whenever $|a|<1$, 
\begin{equation}\label{eq:max} \max_{x\in \partial \Omega(a)}|h(x,a)| \leq M_2 |a|^{m+1}\,.\end{equation}
For let $h(x,y)=\sum_{(i, j)\in A}c_{i,j} x^{i}y^{j}$. By the assumption on $h$ (also recall that $q$ is an integer), we have that $j\geq qi+m+1$ whenever $c_{i,j}\not =0$. 
Furthermore, given $x\in \partial \Omega(a)$, we have $|x|\leq 2|\xi(a)|= 2|a_k| |a|^{-q}$. 
Hence, 
\begin{align*}|h(x,a)|\leq \sum_{(i, j)\in A}|c_{i,j}| |x|^{i}|a|^{j}\leq \sum_{(i, j)\in A}|c_{i,j}| |2 a_k|^{i}|a|^{- qi + j}\\
 \leq |a|^{m+1}\sum_{(i,j)\in A}|c_{i,j}| |2a_k|^{i}\,, \end{align*}
which shows that (\ref{eq:max}) holds.

Inequalities (\ref{eq:min}) and (\ref{eq:max}) imply that for every $a$ with $|a|$ small enough,  $|h(x,a)|< |g(x,a)|$ holds for every $x\in  \partial \Omega(a)$. Rouch\'e's theorem then gives that $g(x,a)$ and $g(x,a)+h(x,a)$ contain the same number of roots in $ \Omega(a)$, counted with multiplicities.  
\end{proof}

We now make some observations about polytopes. 
Let $P\subseteq \R^2$ be a polytope with $k>2$ vertices. Then $P$ has $k$ edges. An edge $e$ will be called a \emph{lower edge}, if $P$ lies in the closed upper-half plane determined by the line passing through $e$.  Similarly, an \emph{upper edge} and the lower-half plane. Every edge $e$ is either a lower or an upper edge, unless $e$ is parallel to the $y$-axis. 
There can be at most two such edges, and we conclude that $P$ has either at least $(k-2)/2$ lower edges, or at least $(k-2)/2$ upper edges. 

Suppose that $P$ has $s$ lower edges $e_1,\dots, e_s$ with gradients $q_1,\dots, q_s$. Then the gradients are distinct and, assuming $q_1<\dots <q_s$,  $P$ contains vertices $(a_1,b_1),\dots, (a_{s+1}, b_{s+1})$ with $a_1<\dots <a_{s+1}$ such that every $e_i$ connects $(a_i,b_{i})$ and $(a_{i+1},b_{i+1})$. Furthermore, the projection of $P$ to the $x$-axis is the interval $[a_1,a_{s+1}]$.  

Let $f(x,y)= \sum_{i,j}a_{i,j}x^iy^j$ and let $e\subseteq \R^2$ be an edge of $\Newt(f)$ connecting vertices $(a,b)$ and $(a',b')$ with $a<a'$. 
We define 
\[ f_e(x,y):= \sum_{(i,j)\in e} a_{i,j} x^{i}y^{j}\,,\,\, f^*_e(x,y):= x^{-a} f_e(x,y)\,. \]

\begin{lemma}\label{lem:afterthought}
Let $L$ be the set of lower edges of $\Newt(f(x,y))$. Let $f^*(x,y):= \prod_{e\in L} f_e^*(x,y)$. 
Then for every $\phi\in [0,2\pi)$ and every $r>0$ sufficiently small, \[|B(f(x,re^{\imath \phi}))- B(f^*(x,e^{\imath \phi}))|<1\,.\]
\end{lemma}
 
\begin{proof} Without loss of generality, assume that $x$ does not divide $f$ and that the lower edges have integer gradients. Otherwise, we can divide by $x$ and replace $y$ by a suitable power of $y$. For $a\in \C\setminus \{0\}$, $\Arg(a)$ denotes the unique $\theta \in (-\pi,\pi]$ with $ a= |a|e^{\imath \theta}$. 

Let $n:= \deg(f(x,y))$, where from now on $\deg$ will denote the degree with respect to the \emph{$x$-variable}.
Let $e_1,\dots, e_{s}$ be the lower edges of $\Newt(f(x,y))$ with gradients $q_1 <\dots <q_s$. As in the above discussion, we have vertices $(a_1,b_1),\dots, (a_{s+1}, b_{s+1})$ so that $e_i$ connects $(a_i,b_{i})$ and $(a_{i+1},b_{i+1})$, and $a_1=0, a_{s+1}=n$. Observe that 
$x$ does not divide $f^*_{e_i}$ and $\deg(f^*_{e_i})= a_{i+1}-a_i$. Hence,
\[\deg(f^*)= \sum_{i=1}^s \deg(f^*_i)= \sum_{i=1}^s (a_{i+1}-a_i)= a_{s+1}-a_0= 
n\,. \]

Let $\phi$ be given and $\epsilon>0$ be sufficiently small. 
We claim the following: given $r$ sufficiently small, 
there is a bijection between the roots of $f^*(x,e^{\imath \phi})$ and  $f(x,re^{\imath \phi})$, so that any two corresponding roots $\xi$, $\xi'$  
satisfy $|\Arg(\xi'/\xi)| < \epsilon$.  
We will call such a bijection an \emph{$\epsilon$-matching}.
This clearly implies that $B(f(x,re^{\imath \phi}))$ and $B(f^*(x,e^{\imath \phi}))$ get arbitrarily close to each other (for our purposes, it is enough to set $\epsilon=1/2n$).

The claim is an application of Lemma \ref{lem:rouche}. Let $d_i:= \deg(f^*_{e_i})$.  
 We can factor each $f_{e_i}^*$ as  
\begin{equation}\label{eq:factor}f_{e_i}^*(x,y)= c_0 y^{b_i}\prod_{j=1}^{d_i}(xy^{q_i}- c_{j})\,, \end{equation}
This means that the roots of $f_{e_i}^*(x,re^{\imath \phi})$ 
lie in the disc    
\[D_i(r)=  \{ z\in \C: m_i r^{-q_i} \leq |z|\leq m_i' r^{-q_i} \}\,,\]
where $0< m_i\leq m_i'$ are independent of $r$. 
By the definition of $f_{e_i}$, the Newton polytope of $f(x,y)-f_{e_i}(x,y)$ lies in the strict upper-half plane determined by the edge $e_i$. From Lemma \ref{lem:rouche}, we conclude that for every $r>0$ small enough, $f(x,re^{\imath \phi})$ contains at least $d_i$ roots 
in the disc 
\[D_i'(r)= \{ z\in \C: m_i(1-\epsilon) r^{-q_i} <|z|\leq m_i' (1+\epsilon)r^{-q_i}\}\,\]
and, moreover, there is an $\epsilon$-matching between $d_i$ of these roots and the roots of $f^*_{e_i}(x,re^{\imath \phi})$.
As $r$ approaches zero, the discs $D_i'(r)$ and $D_j'(r)$ become disjoint for distinct $i$ and $j$. Since $\deg(f)=\deg(f^*)$, this means that there is an $\epsilon$-matching between the roots of $f(x,re^{\imath \phi})$ and  $f^*(x,re^{\imath \phi})$. To conclude the claim, observe  from (\ref{eq:factor}) that the arguments of the roots of $f^*(x,re^{\imath \phi})$ do not depend on $r$.
\end{proof}

\begin{proof}[{\bf Proof of Theorem \ref{thm:Newt}}] Let $f(x,y)$ be such that $\Newt(f(x,y))$ has $k$ vertices. Without loss of generality, we will assume that  $\Newt(f)$ has $s\geq (k-2)/2$ lower edges. For otherwise, take the polynomial $y^m f(x,y^{-1})$ for $m$ sufficiently large.   

Using Lemma \ref{lem:afterthought}, it is enough to show 
 there exists $\phi\in [0,2\pi)$ with \begin{equation}\label{eq:claim}B(f^*(x,e^{\imath \phi}))\geq \Omega(\sqrt{s})\,.\end{equation} 
If the gradients of the lower edges are integers, we can factor $f^*$ as (\ref{eq:poly}) and conclude (\ref{eq:claim}) from Theorem \ref{thm:run0} as in the discussion in  Section \ref{sec:connection}. If the gradients are not integers, take instead $f^*(x,y^m)$ for a suitable $m$. 
\end{proof}

\subsection{Proof of Theorem \ref{thm:cor}}\label{sec:cor}
 As already remarked, a weaker version of Theorem \ref{thm:cor} follows from Theorem \ref{thm:Newt}. However, the following proof of the full version is self-contained, though similar to that of Theorem \ref{thm:Newt}.   

An analogue of Lemma \ref{lem:rouche} is the following: 
\begin{lemma}\label{lem:rouche2} Let $g(x,y)$ be a \emph{real} polynomial as in (\ref{eq:g}) with $c_{n_1}c_{n_2}<0$ and $q\in \Q$. Let $h(x,y)$ be a real polynomial such that $\Newt(h)$ lies strictly above the line $\ell$. Then there exist $0<d<d'$ such that for every $0< r$ sufficiently small, $g(x,r)+h(x,r)$ contains a root in the interval $(dr^{-q},d'r^{-q})$. 
\end{lemma}

\begin{proof} The assumption $c_{n_1}c_{n_2}<0$ guarantees the existence of $0<d<d'$ with $g(d,1)g(d',1)<0$. This means that for every $0<r$, also $g(dr^{-q},r)g(d'r^{-q},r)<0$. 
As in the proof of Lemma \ref{lem:rouche}, it can be shown that for $r>0$ sufficiently small,  $|h(dr^{-q},r)|< |g(dr^{-q},r)|$ and $|h(d'r^{-q},r)|< |g(d'r^{-q},r)|$. This shows that 
$g(x,r)+h(x,r)$ has different signs  on the endpoints of the interval  $[dr^{-k},d'r^{-k}]$, and the interval must contain a real root. 
\end{proof}

The following lemma is a substitute for Theorem \ref{thm:run0}. 
Let $r_1,\dots, r_k$ be a sequence of real numbers. We define
\[V(r_1,\dots, r_k):= |\{i\in [ k-1]: r_ir_{i+1}<0\}|\,,\] 
 the number of sign variations in the sequence.
 
 \begin{lemma}\label{lem:cos} Let $\alpha_1,\dots, \alpha_k\in [0,2\pi)$ and let $n_1,\dots, n_k$ be positive integers such that $n_i\not=n_{i+1}$ for every $i\in [k-1]$. Then there exists $\phi\in [0,2\pi)$ such that\footnote{We are not trying to optimize the constant; a different argument would give an improvement of $(k-1)/6$.}
 $V(\cos(\alpha_1+\phi n_1),\dots, \cos(\alpha_k+\phi n_k))\geq (k-1)/8$.
 \end{lemma}
 \begin{proof} Pick a random $x\in [0,1]$. Let $f_{i}(x):= \cos(\alpha_i+ 2\pi  n_i x) \cos(\alpha_{i+1}+ 2\pi  n_{i+1} x)$. 
 We claim that \[ \Pr[f_i(x)< 0]\geq 1/8\,.\]
This can be seen as follows. An easy calculation shows that $\int_0^1 f_i(x)dx=0$ and $\int_0^1 f_i(x)^2dx=1/4$. 
Let $A:= \{x\in [0,1]: f_i(x)<0\}$. Then 
\begin{align*} \frac{1}{4} =& \int_0^1 f_i(x)^2dx \\ \leq &\int_0^1 |f_i(x)|dx= \int_{A} |f_i(x)|dx+ \int_{[0,1]\setminus A} |f_i(x)|dx= 2 \int_{A} |f_i(x)|dx \leq 2|A|\,, 
\end{align*}
 which shows that $|A|\geq 1/8$. 
 
 Let $\chi_i(x)\in \{0,1\}$ be the indicator function of the event that $f_i(x)< 0$. Then the expectation of $\chi_1(x)+\dots+\chi_{k-1}(x)$ is at least $(k-1)/8$. Hence there exists $x\in [0,1)$ so that 
 $f_i(x)<0$ holds for at least $(k-1)/8$ of the $i$'s.  
  \end{proof}

 \begin{proof}[\bf Proof of Theorem \ref{thm:cor}] 
 As in the proof of Theorem \ref{thm:Newt}, we can assume that the polytope has $s\geq (k-2)/2$ lower edges. Let $e_1,\dots,e_s$ be such edges with gradients $q_1,\dots, q_s$, where $e_i$ connects 
 $(a_i,b_i)$ and $(a_{i+1},b_{i+1})$ with $a_i<a_{i+1}$.  Let $r_i e^{\imath  \alpha_i}$ be the coefficient of $x^{a_i}y^{b_i}$ in $f$ (where $r_i>0$, $\alpha_i\in [0,2\pi)$).

 Given $a= r e^{\imath\phi}$, we can write $ \re(f_{e_i}(x, a))$ as
 \[ 
 \cos({\alpha_i+ b_i \phi}) r_{i}r^{b_i}  x^{a_i} + \cos(\alpha_{i+1}+ b_{i+1}\phi)r_{i+1}r^{b_{i+1}}  x^{a_{i+1}}+ \re(u_i(x,a)) \,,
 \]
 where $u_i(x,y)$ is a polynomial such that $\Newt(u_i(x,y))$ lies on  the line strictly between the points $(a_i,b_i)$ and $(a_{i+1},b_{i+1})$. 
 Let $T(\phi)$ be the sequence $\cos(\alpha_{1}+\phi b_1),\dots, \cos(\alpha_{s+1}+\phi b_{s+1})$.
  Note that $b_i\not=b_{i+1}$ holds for every $i\in [s]$, with at most one exception. By the previous lemma, there exists a $\phi$ such that $V(T(\phi))\geq (s-1)/8$ (this ''one exception" compensated by the sequence having length $s+1$). Fix such a $\phi$. 
 Given an $i$ with $\cos(\alpha_{i}+\phi b_i)\cos(\alpha_{i+1}+\phi b_{i+1})< 0$, we can apply 
  Lemma \ref{lem:rouche2}, to conclude that $\re(f(x,re^{\imath \phi}))$ has a root in the interval $(d_i r^{-q_i}, d_i' r^{-q_i})$ for every $r$ sufficiently small.  As $r$ approaches zero, the intervals corresponding to different $i$'s are disjoint (the gradients $q_i$ are distinct).  This gives that $\re(f(x,re^{\imath \phi}))$ has at least $(s-1)/8$ distinct real roots for $r$ sufficiently small.   
 \end{proof}
 
 \paragraph{Acknowledgement} The author thanks B. Green, P. Pudl\'ak, M. Rojas, and the anonymous referees for their comments.

 \bibliographystyle{plain}

\end{document}